\documentclass[journal]{IEEEtran}
\usepackage{amsfonts}
\usepackage{cite,graphicx,amsmath,amsthm}
\usepackage{subfigure}
\usepackage{fancyhdr}
\usepackage{dsfont}
\usepackage{array,color}
\usepackage{bm}
\usepackage{float}
\usepackage{algorithm}
\usepackage{algpseudocode}
\usepackage{multirow}
\usepackage{booktabs}
\usepackage{makecell}
\allowdisplaybreaks[4]

\newtheorem{lemma}{Lemma}

\newtheorem{proposition}{Proposition}

\newtheorem{corollary}{Corollary}

\newtheorem{property}{Property}

\newtheorem{remark}{Remark}

\newtheorem{claim}{Claim}

\begin{document}

\title{\huge Multi-Point Integrated Sensing and Communication: \\Fusion Model and Functionality Selection}

\author{Guoliang Li, Shuai Wang, Kejiang Ye, Miaowen Wen,\\Derrick Wing Kwan Ng, \emph{Fellow, IEEE}, and Marco Di Renzo, \emph{Fellow, IEEE}   % <-this % stops a space
\vspace{-24pt}
\thanks{
This work was supported by the National Key R\&D Program of China (No. 2021YFB3300200), the National Natural
Science Foundation of China (No. 62001203), the Shenzhen Science and Technology Program (No. RCB20200714114956153), and the Guangdong Basic and Applied Basic Research Project (No. 2021B1515120067). The work of D. W. K. Ng is supported by the Australian Research Council's Discovery Project (DP210102169).
Corresponding author: Shuai Wang.

Guoliang~Li is with the Shenzhen Institute of Advanced Technology, Chinese Academy of Sciences, Shenzhen 518055, China, and also with the Department of Electrical and Electronic Engineering, Southern University of Science and Technology (SUSTech), Shenzhen 518055, China (e-mail: ligl2020@mail.sustech.edu.cn).

Shuai~Wang and Kejiang~Ye are with the Shenzhen Institute of Advanced Technology, Chinese Academy of Sciences, Shenzhen 518055, China (e-mail: \{s.wang,~kj.ye\}@siat.ac.cn).

Miaowen Wen is with the School of Electronic and Information Engineering, South China University of Technology, Guangzhou 510640, China (e-mail: eemwwen@scut.edu.cn).

Derrick Wing Kwan Ng is with the School of Electrical Engineering and Telecommunications, the University of New South Wales, Australia (e-mail: w.k.ng@unsw.edu.au).

Marco Di Renzo is with Universit\'e Paris-Saclay, CNRS, CentraleSup\'elec, Laboratoire des Signaux et Syst\`emes, 3 Rue Joliot-Curie, 91192 Gif-sur-Yvette, France (marco.di-renzo@universite-paris-saclay.fr).
}
}
\maketitle

\begin{abstract}
Integrated sensing and communication (ISAC) represents a paradigm shift, where previously competing wireless transmissions are jointly designed to operate in harmony via the shared use of the hardware platform for improving the spectral and energy efficiencies.
However, due to adversarial factors such as fading and interference, ISAC may suffer from high sensing uncertainties.
This paper presents a multi-point ISAC (MPISAC) system that fuses the outputs from multiple ISAC devices for achieving higher sensing performance by exploiting multi-view data redundancy.
Furthermore, we propose to effectively explore the performance trade-off between sensing and communication via a functionality selection module that adaptively determines the working state (i.e., sensing or communication) of an ISAC device.
The crux of our approach is to derive a fusion model that predicts the fusion accuracy via hypothesis testing and optimal voting analysis.
Simulation results demonstrate the superiority of MPISAC over various benchmark schemes and show that the proposed approach can effectively span the trade-off region in ISAC systems.
\end{abstract}

\begin{IEEEkeywords}
Integrated sensing and communication, multi-view fusion, functionality selection.
\end{IEEEkeywords}

\IEEEpeerreviewmaketitle

\vspace{-7pt}

\section{Introduction}

\begin{figure*}[!t]
\centering
\includegraphics[width=0.98\textwidth]{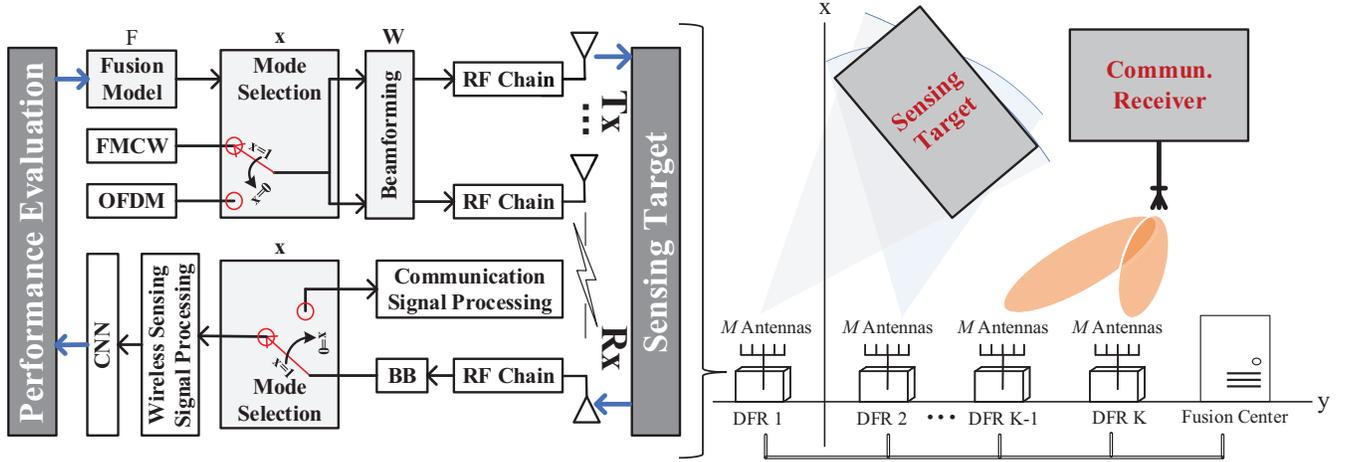}
\caption{The MPISAC system with $K$ DFRs, one target object, and one communication receiver.
The architecture of each DFR is shown on the left-hand side.
A functionality selection module controls the transmitted symbol and the receiving process under the instruction of the fusion model.
Note that the system could be extended to multi-target scenarios since
the distance and azimuth of each target could be extracted via time of flight estimation.} \label{Scenario}
\vspace{-12pt}
\end{figure*}

\IEEEPARstart{E}{lectromagnetic} waves can sense environments and carry information simultaneously.
Nevertheless, the two functionalities have traditionally been studied separately, resulting in resource competition between sensing and communication systems.
Currently, wireless systems experience a paradigm shift towards integrated sensing and communication (ISAC) \cite{ISAC}, i.e., unifying wireless sensing of environments and transmission of data so as to make the best use of the limited spectrum and costly hardware {platforms}.

Despite various efforts and successes in developing ISAC techniques \cite{Fan_1,Fan_2}, a number of technical
challenges still need to be properly handled, which include:
\begin{itemize}
\item[1)] \textbf{Reduction of sensing uncertainties}. The sensing data at each individual radar may be noisy or incomplete due to wireless fading and interference \cite{Noise}. Aggregating redundant measurements from multiple radars can help mitigate such uncertainties \cite{Fusion_1,Fusion_2}.

\item[2)] \textbf{Binary functionality selection}. The dual-functional radar (DFR) contains discrete elements such as switches \cite{Fan_1}. The binary state of the switch at each DFR needs to be carefully controlled according to both the application-layer (e.g., detection error \cite{Noise}) and the physical-layer parameters (e.g., channel qualities).

\item[3)] \textbf{Characterizing sensing-communication trade-off}. For conventional geometric sensing tasks (e.g., localization), the sensing performance is known to be evaluated by the Cram\'er-Rao bound (CRB) \cite{Fan_2}. Nonetheless, for the emerging semantic and prediction sensing tasks, the sensing performance cannot be straightforwardly computed as sophisticated deep neural networks are involved for feature extraction \cite{SPAWC}.
\end{itemize}
Existing works address the above issues from either an ISAC (e.g., \cite{Fan_1,Fan_2}) or multi-radar fusion (e.g., \cite{Fusion_1,Fusion_2}) perspective, which ignores the interdependency between sensing and communication in multi-point scenarios, thus failing in achieving high sensing accuracy and spectral efficiency simultaneously under stringent resource constraints.

To fill this gap, this paper studies a multi-point ISAC (MPISAC) system that fuses the outputs from multiple DFRs under wireless resource constraints.
First, by leveraging a set of newly derived models, the fusion gain brought by the MPISAC over the conventional ISAC is quantified.
Second, by combining the two conflicting goals of fusion accuracy and communication rate, via the weighted-sum method,
a joint transmit beamforming and functionality selection problem is formulated to maximize the mixed objectives.
Despite its implicit objective function and discontinuous design variables,
a hybrid meta-heuristic and optimization (HMO) algorithm maximizing a surrogate objective function is developed, which converges to a local optimal solution to the surrogate problem.
The approximation gap between the original and surrogate problems is quantified.
Third, the proposed algorithm is implemented in a high-fidelity wireless radar simulator.
Experimental results show that the quality of the wireless channel plays a key role in functionality selection and that incorporating fusion mechanisms improves the overall ISAC system performance.
The performance trade-off between sensing and communication is also illustrated.
To the best of our knowledge, this is the first work that integrates multi-view fusion and ISAC technologies to facilitate the design of wireless {systems}.

\section{System Model}

We consider the {MPISAC} system shown in Fig.~\ref{Scenario}, where $K$ DFRs, each equipped with $M$ transmit antennas and a single receive antenna, adopt electromagnetic waves for anomaly target detection while transmitting a comman information message to a single-antenna mobile device. The DFRs are wireline connected to a fusion center (i.e., edge server), which aggregates the outputs of multiple DFRs for achieving higher detection accuracy.
The left-hand side of Fig.~\ref{Scenario} illustrates the architecture of each DFR.
It can be seen that both the sensing and communication functionalities share the same hardware units, e.g., waveform generators, beamformers, radio-frequency (RF) chains, etc.
A functionality selection module, which is represented by a binary variable $\mathbf{x}=[x_1,\cdots,x_K]^T\in\{0,1\}^K$, is adopted to determine the working states of all the DFRs. Specifically, $x_{i}=1,~\forall i\in \left\{1,...,K\right\}$, represents that the $i$-DFR operates in the sensing mode and $x_{i}=0$ denotes that the $i$-DFR operates in the communication mode.

\subsubsection{Sensing Signal Model}
When $x_i=1$, the $i$-th DFR sends a probing signal, which is reflected by the target and is received by the $i$-th DFR as
\begin{align}
	\vspace{-7pt}
	\!\!\!
	y_{i}&=
	x_{i}\mathbf{g}_{ii}^{H}\mathbf{w}_{i}s_i+
	\sum\limits_{j\neq i}I_{j,i}
	+z_i, \label{y_s}
	\vspace{-7pt}
\end{align}
where $I_{j,i}=x_{j}\mathbf{h}_{ji}^{H}\mathbf{w}_{j}s_j+(1-x_{j})\mathbf{h}_{ji}^{H}\mathbf{w}_{j}c$.
In particular, $s_i\in \mathbb{C}$ and $c\in \mathbb{C}$ are the sensing and communication signals, respectively, where $\mathbb{E}\left[|s_i|^2\right]=1$, $\forall i$, and $\mathbb{E}\left[|c|^2\right]=1$, without loss of generality.
The vector $\mathbf{w}_{i} \in \mathbb{C}^M$ denotes the transmit beamformer at the $i$-th DFR with the power constraints $\Vert\mathbf{w}_i\Vert^2\leq P_T$ and $\sum_{i=1}^{K}\Vert\mathbf{w}_i\Vert^2\leq P_{\rm{sum}}$.
The vector $\mathbf{g}_{ji}\in \mathbb{C}^{M}$ denotes the two-hop channel that is the product of the channel from the $j$-th DFR to the target and the channel from the target to the $i$-th DFR.
The vector $\mathbf{h}_{ji}\in \mathbb{C}^{M}$ represents the line-of-sight (LOS) channel from the $j$-th DFR to the $i$-th DFR.
The scalar $z_i \sim \mathcal{CN}(0,\sigma^2)$ is the additive white Gaussian noise (AWGN) and $\sigma^2$ is the noise power.
The associated sensing signal-to-interference-plus-noise ratio (SINR) at the $i$-th DFR is\footnote{Note that the communication signals received at the $i$-th DFR
are identical but they go through different channels.
Hence, we need to sum up their amplitudes instead of their signal strengths for computing the total power.}
\begin{align}
&\text{SINR}_i^{\mathrm{s}}=\frac{x_{i}|\mathbf{g}_{ii}^{H}\mathbf{w}_i|^2}
{\sigma^2 + \sum_{j\neq i}|x_j\mathbf{h}_{ji}^{H}\mathbf{w}_{j}|^2 + |\sum_{j\neq i}(1-x_j)\mathbf{h}_{ji}^{H}\mathbf{w}_{j}|^2}.
\end{align}

\subsubsection{Anomaly Detection Model}
To extract the target information embedded in the signal $y_{i}$, we feed $y_{i}$ into a sensing signal processing module (as shown in the lower left of Fig.~\ref{Scenario}) for generating a motion-related image termed ``spectrogram'' (as shown in Fig.~\ref{SINR}) \cite{SPAWC}. As such, the original radar-based anomaly detection problem is converted into an image classification problem that can be effectively tackled by {a} convolutional neural network (CNN) \cite{Noise}.
However, the image quality could be unacceptable if the SINR is below a certain threshold as observed from Fig.~\ref{SINR}.
Furthermore, the detection accuracy of {a} CNN is significantly deteriorated if the image is of low quality, as shown in Fig.~\ref{Acc}. Lastly, forwarding false detections to the fusion center may break down the entire MPISAC system.
Therefore, among all the sensing DFRs $\{i:x_i=1\}$, it is necessary to ignore ineffective DFRs (i.e., those with low detection accuracies) and to fuse only effective DFRs (i.e., those with high detection accuracies).
This can be realized by setting a target detection accuracy threshold and reading the associated SINR threshold $\gamma$ from Fig.~\ref{Acc}
(we set $\gamma=30\,$dB according to the widely adopted detection accuracy threshold 89\% \cite{Noise} corresponding to the gray bar in Fig. 2(b)).
Since the detection accuracy is monotonically increasing with the received SINR as shown in Fig.~\ref{Acc}, the $i$-th DFR is deemed effective if $\text{SINR}_i^{\mathrm{s}}\geq\gamma$; otherwise, the $i$-th DFR is deemed ineffective.
The set of effective DFRs is thus $\mathcal{E}=\{i: x_i=1,\, \text{SINR}_i^{\mathrm{s}}\geq\gamma\}$ whose cardinality is $|\mathcal{E}|$.

\begin{figure*}[!t]
	\vspace{-0.2cm}
	\centering
	\subfigure[]{
		\label{SINR} %% label for second subfigure width=3.34in,height=1.02in
		\includegraphics[height=2.6in]{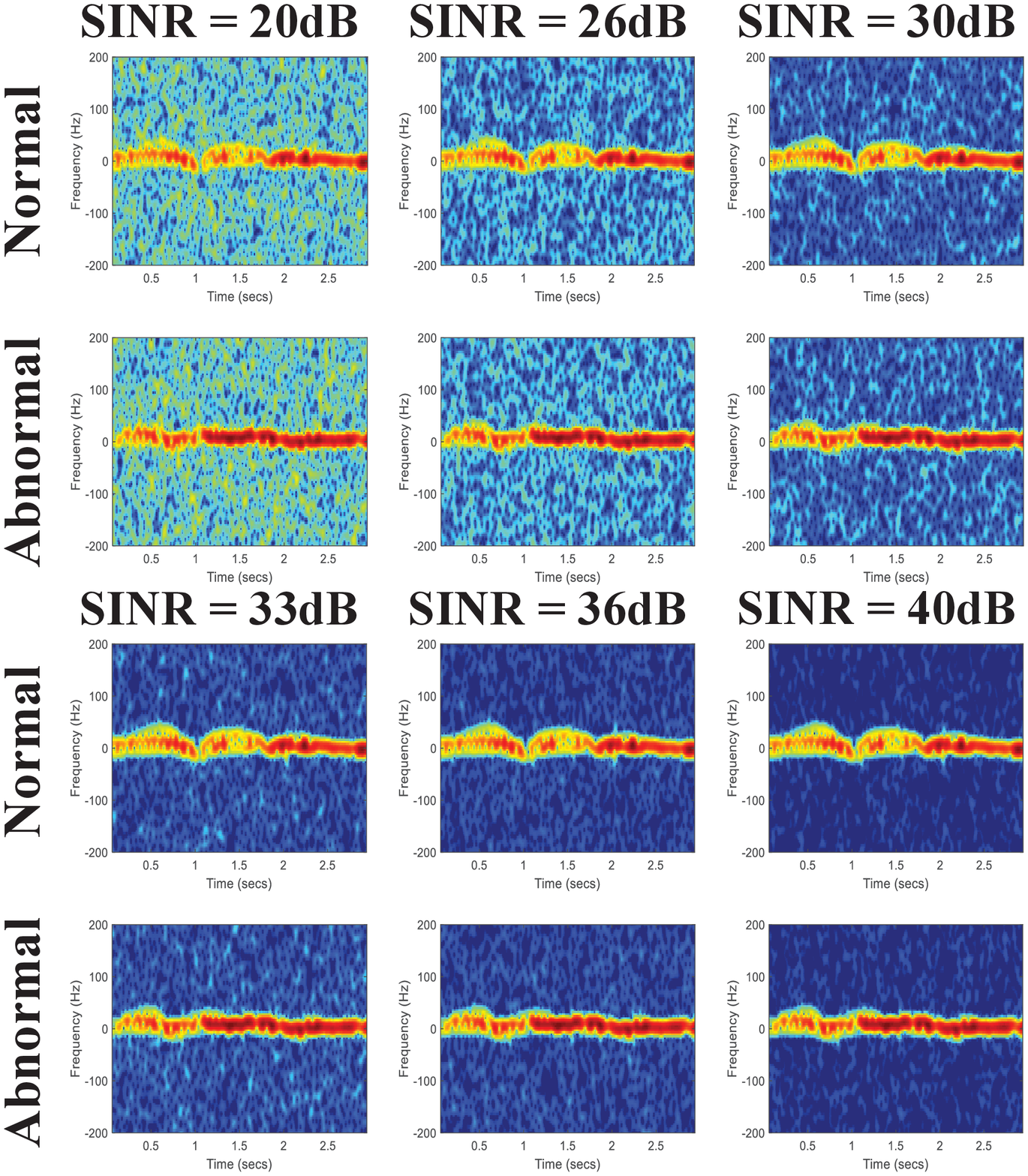}}
		\subfigure[]{
		\label{Acc} %% label for second subfigure width=1.6in,height=1.43in
		\includegraphics[height=2.6in]{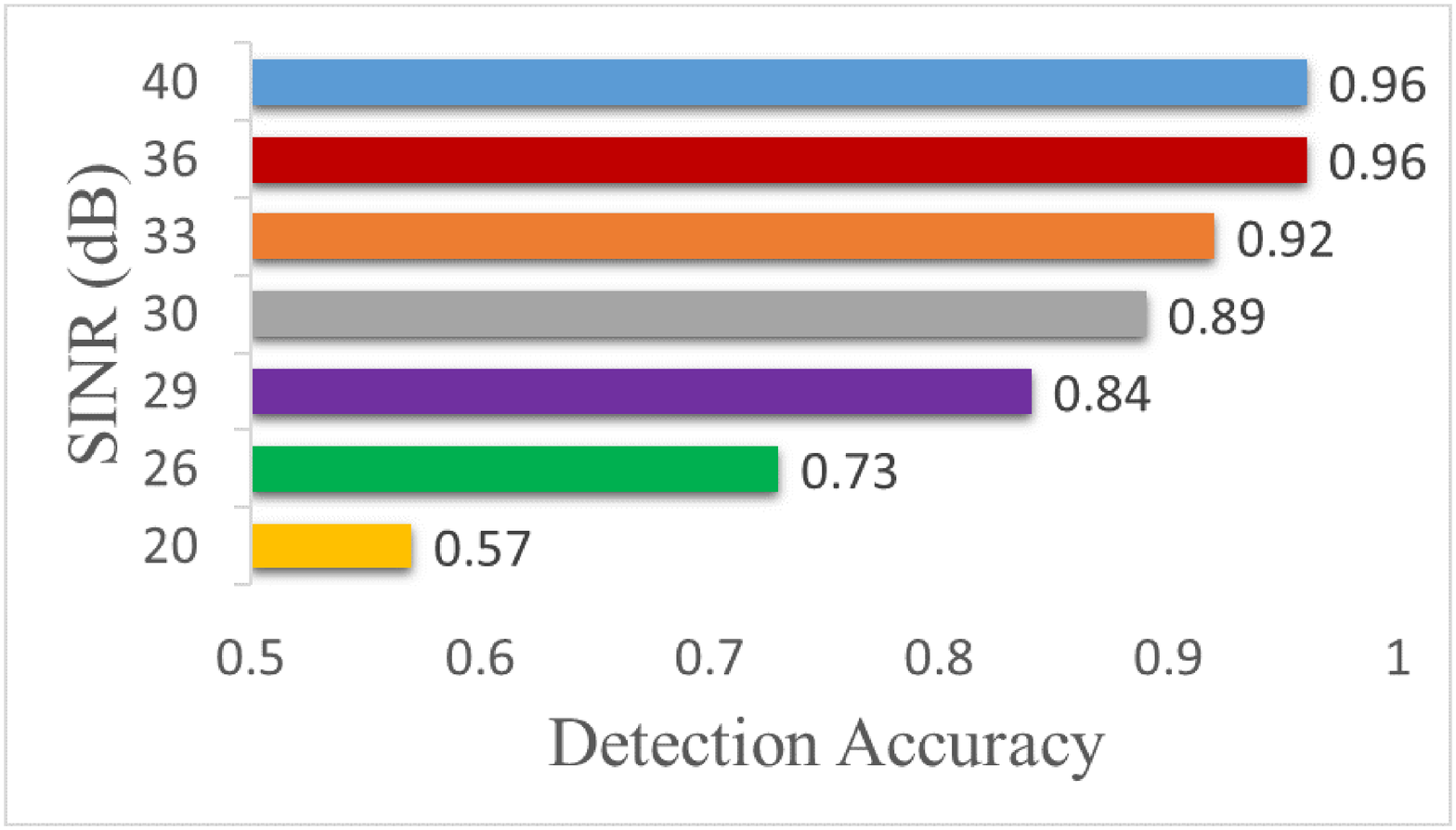}}
		\subfigure[]{
		\label{Optimla_n} %% label for second subfigure width=1.6in,height=1.43in
		\includegraphics[height=2.6in]{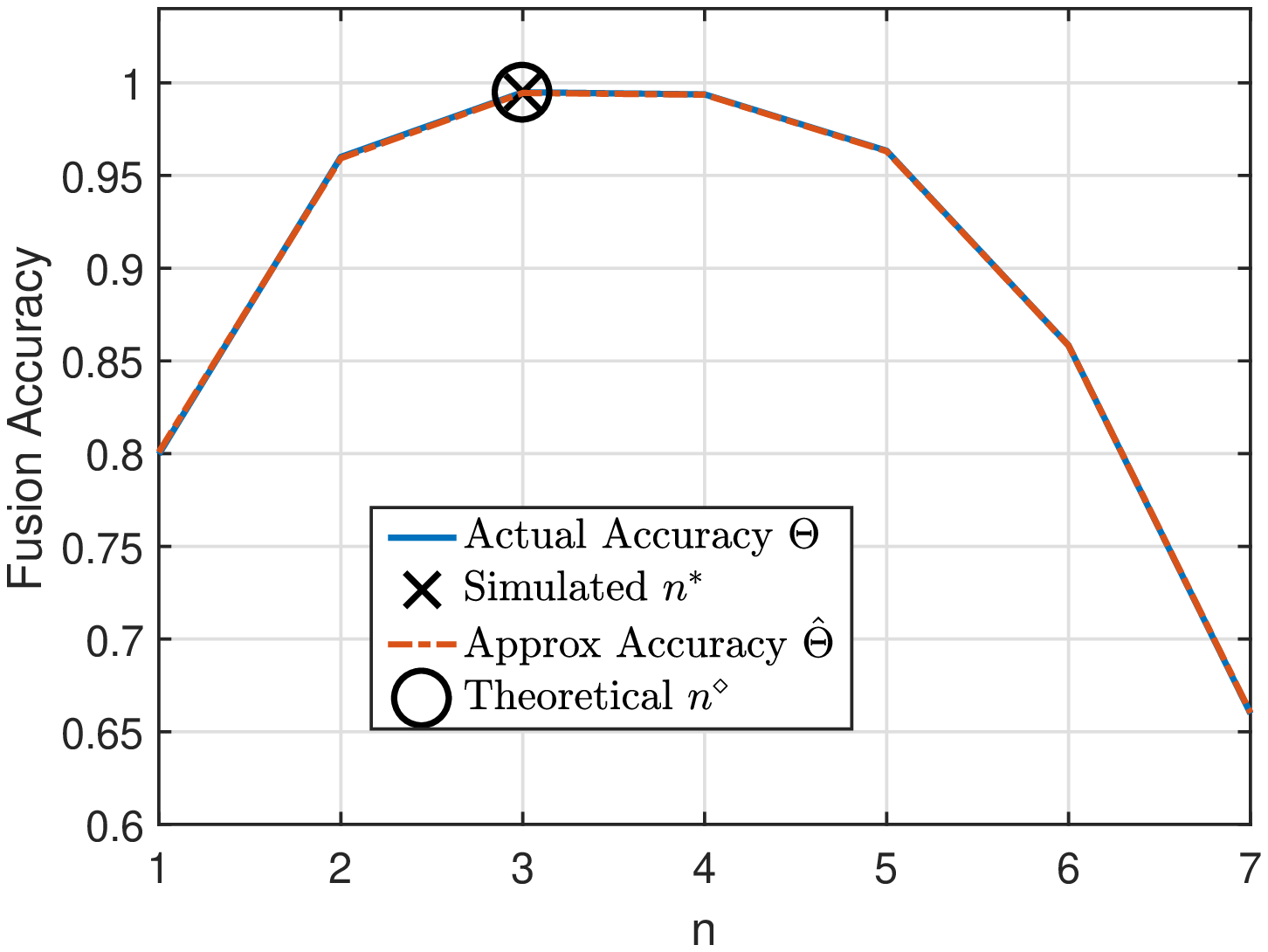}}
		\subfigure[]{
		\label{converge} %% label for second subfigure width=1.6in,height=1.43in
		\includegraphics[height=2.6in]{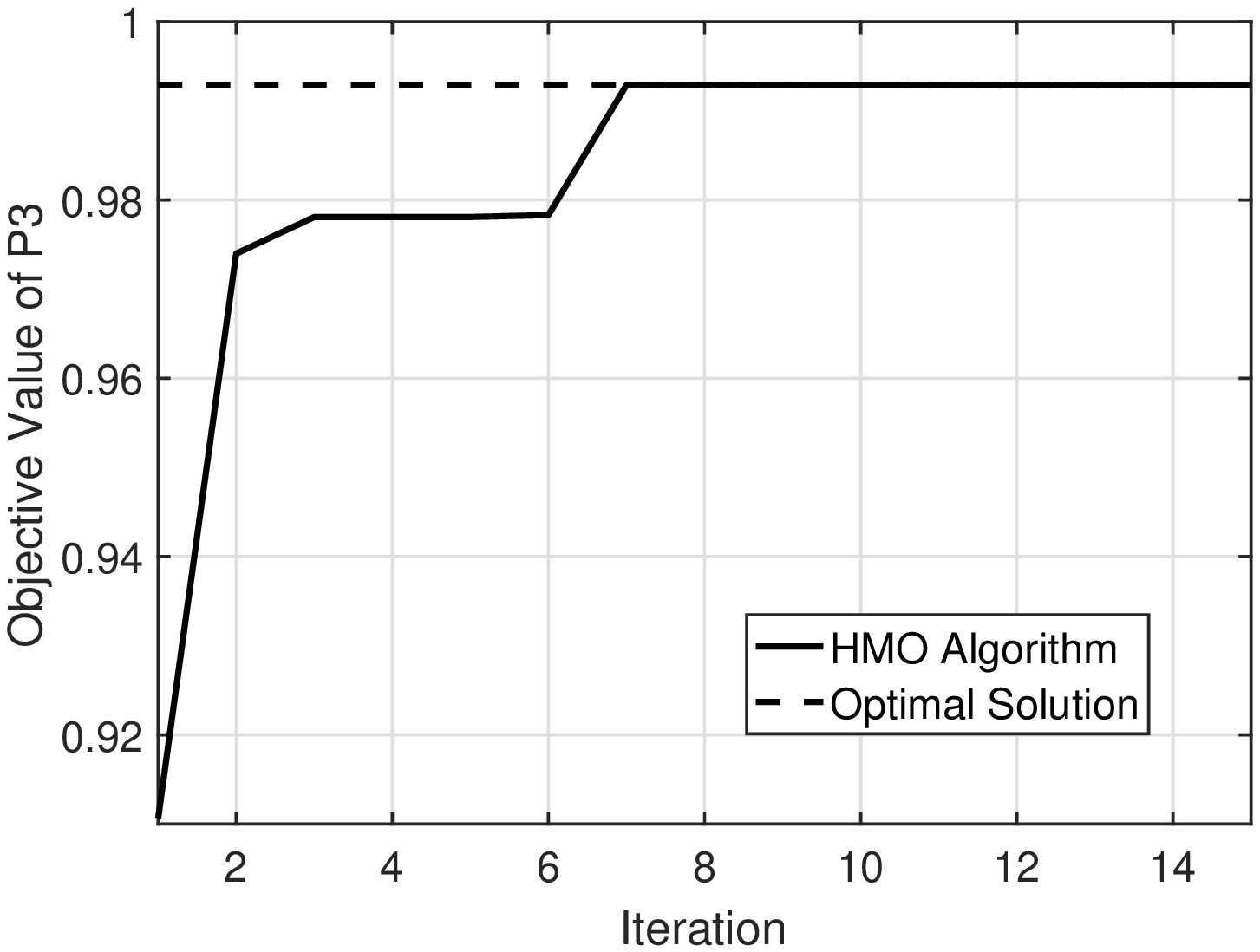}}
	\caption{ (a) Human-motion images under different SINRs; (b) Detection accuracy versus SINR; (c) Actual and approximate fusion accuracy versus $n$, where the associated false negative and false positive rates are $\left\{P_i\right\}_{i=1}^{7} = \{0.05,0.04,0.07,0.02,0.03,0.08,0.10\}$ and $\left\{Q_i\right\}_{i=1}^{7} = \{0.19,0.21,0.17,0.16,0.15,0.13,0.11\}$, respectively; (d) The objective value of $\mathcal{P}3$ versus the number of iteration for the proposed HMO algorithm.}
	\label{plot} %% label for entire figure
\vspace{-12pt}
\end{figure*}

\subsubsection{Fusion Model}
The outputs of {the} effective DFRs $\mathcal{E}$ are aggregated via voting, which can be modeled by hypothesis-testing:
\begin{align}
H_1: \mathrm{Normal~target}; \quad H_2: \mathrm{Abnormal~target}.
\end{align}
Then, the ``$n$ out of $|\mathcal{E}|$'' voting rule \cite{voting} is given by
\begin{align}
\mathcal{H}_1:\sum_{i\in\mathcal{E}}D_i< n; \quad \mathcal{H}_2:\sum_{i\in\mathcal{E}}D_i\geq n,
\end{align}
where $D_i$ is the binary inference result of the $i$-th DFR with $D_i=0$ standing for a normal target and $D_i=1$ standing for an abnormal target.
The fusion accuracy with the voting threshold $n$ is\footnote{Due to the different observation views,
the sensing data at different DFRs is assumed to be independent \cite{Fusion_1}.}
\vspace{-0.3cm}
\begin{align}
	\Theta(\mathcal{E}|n)&=\frac{1}{2}\sum\limits_{l=0}^{n-1}\sum\limits_{F\in\mathcal{F}_l}\prod
	_{i\in F}P_i\prod_{i\in F^c}(1-P_i)
\nonumber\\
&\quad{}{}
+\frac{1}{2}\sum\limits_{l=0}^{|\mathcal{E}|-n}\sum\limits_{F\in\mathcal{F}_l}\prod
	_{i\in F}Q_i\prod_{i\in F^c}(1-Q_i), \label{exact}
\end{align}
where the parameters in \eqref{exact} are detailed as follows:
(1) $P_i$ and $Q_i$ are the false negative (false alarm) and false positive (missing alarm) rates at the $i$-th DFR, respectively, which are estimated from the experimental data; note that different DFRs may have different $(P_i,Q_i)$ under the same SINR due to different observation angles;
(2) $\mathcal{F}_l$ contains all the subsets with $l$ unique DFRs from $\mathcal{E}$ and its cardinality is $|\mathcal{F}_l|=\tbinom{|\mathcal{E}|}{l}=\frac{|\mathcal{E}|!}{l!\left(|\mathcal{E}|-l\right)!}$;
(3) For any $F\in \mathcal{F}_l$, its complement is denoted as $F^{c}=\mathcal{E}\backslash F$.
It can be seen that the fusion accuracy under optimal voting is $\max_n\Theta(\mathcal{E}|n)$.

\subsubsection{Joint Transmission Model}
For the DFRs with $\{x_i=0,\forall i\}$, the jointly transmitted signal received at the communication receiver is
\begin{align}
	r&=\sum_{i=1}^K(1-x_{i})\mathbf{f}_{i}^{H}\mathbf{w}_{i}c+
		\sum_j I_{j}'
		+z_i', \label{y_c}
\end{align}
where $I_{j}'=x_{j}\mathbf{f}^{H}_{j}\mathbf{w}_{j}s_j$.
{The vector} $\mathbf{f}_{j}\in \mathbb{C}^{M}$ is the channel from the $j$-th DFR to the communication receiver.
{The scalar} $z_i' \sim \mathcal{CN}(0,\sigma^2)$ is the additive white Gaussian noise (AWGN).
The associated communication SINR is written as
\begin{align}
\text{SINR}^{\mathrm{c}}=\frac{|\sum_{i=1}^K(1-x_{i})\mathbf{f}_i^{H}\mathbf{w}_i|^2}
{\sigma^2 + \sum_{j} |x_{j}\mathbf{f}^{H}_{j}\mathbf{w}_{j}|^2}.
\end{align}
The communication spectral efficiency in bps/Hz is
\begin{align}\label{drate}
R(\mathbf{x},\{\mathbf{w}_i\})=\log_{2}\left(1+\frac{|\sum_{i=1}^K(1-x_{i})\mathbf{f}_i^{H}\mathbf{w}_i|^2}
{\sigma^2 + \sum_{j} |x_{j}\mathbf{f}^{H}_{j}\mathbf{w}_{j}|^2}\right).
\end{align}

\vspace{-0.3cm}
\section{Proposed Functionality Selection and Beamforming Design Algorithm}

The MPISAC system aims to maximize the sensing and communication performance under the mode selection, transmit power, and effective sensing constraints, which results in the following multi-objective optimization problem
\begin{subequations}
	\begin{align}
	\mathcal{P}0:\mathop{\mathrm{max}}_{\substack{\mathbf{x},\{\mathbf{w}_{i}\},\mathcal{E}}}~~
	&
	\left\{\max_n\Theta(\mathcal{E}|n),R(\mathbf{x},\{\mathbf{w}_i\})\right\} \nonumber\\
	\mathrm{s.t.}~~~ \label{P101}
	&x_i \in \left\{0,1\right\},~\forall i, \\
	&\sum_{i=1}^{K}\Vert\mathbf{w}_i\Vert^2\leq P_{\rm{sum}},~\Vert\mathbf{w}_i\Vert^2\leq P_T,\,\forall i, \label{P102}
\\ &\mathcal{E}=\{i: x_i=1,\, \text{SINR}_i^{\mathrm{s}}\geq\gamma\}.
	\end{align}
\end{subequations}
To {solve} $\mathcal{P}0$, the following transformations are adopted:
\begin{itemize}
\item[1)] The multi-objective function in $\mathcal{P}0$ is transformed into $(1-\mu)\max_n\Theta(\mathcal{E}|n)+\mu R(\mathbf{x},\{\mathbf{w}_i\})$, where $\mu\in[0,1]$ is the weight to prioritize the two objectives.\footnote{The optimal solution to the weighted-sum problem of $\mathcal{P}0$ is guaranteed to be a Pareto optimal solution to $\mathcal{P}0$ \cite{Weight}.}
\item[2)] The optimal $\mathbf{x}^*$ and $\mathcal{E}^*$ satisfy $\mathcal{E}^*=\{i:x_i^*=1\}$.\footnote{This can be proved by contradiction.
Suppose that there exists some $j\in\{i:x_i^*=1\}$ such that $x_j^*=1$ and $\text{SINR}_j^{\mathrm{s}}<\gamma$.
Then, we can always set $x_j=0$ such that the value of $\max_n\Theta(\mathcal{E}|n)$ is unchanged while $R(\mathbf{x},\{\mathbf{w}_i\})$ is increased.
This contradicts the optimality of $\mathbf{x}^*$.
Hence, for any $j\in\{i:x_i^*=1\}$, we must have $\text{SINR}_j^{\mathrm{s}}\geq\gamma$, meaning that $\mathcal{E}^*=\{i:x_i^*=1\}$.}
Therefore, replacing $\mathcal{E}$ with a new variable $\mathcal{S}=\{i:x_i=1\}$ (with its cardinality $|\mathcal{S}|$ being the number of sensing DFRs) in $\mathcal{P}0$ would not change the problem solution.
\end{itemize}
With these observations, $\mathcal{P}0$ is \emph{equivalently} transformed into
\begin{subequations}
	\begin{align}
	\mathcal{P}1:\mathop{\mathrm{max}}_{\substack{\mathbf{x},\{\mathbf{w}_{i}\},\mathcal{S}}}~~
	&
	(1-\mu)\max_n\Theta(\mathcal{S}|n)+\mu R(\mathbf{x},\{\mathbf{w}_i\})\nonumber\\
	\mathrm{s.t.}~~~ &\eqref{P101},\,\eqref{P102},\,\mathcal{S}=\{i:x_i=1\},
\\ &\text{SINR}_i^{\mathrm{s}}\geq\gamma,~\forall i\in\mathcal{S}.
	\end{align}
\end{subequations}
The challenges of solving problem $\mathcal{P}1$ are three-fold: 1) the nonlinear coupling between $\mathbf{x}$ and $\{\mathbf{w}_{i}\}$; 2) the implicit function $\max_n\Theta(\mathcal{S}|n)$; % due to the factorials in (\ref{exact});
3) the discontinuity of $\mathbf{x}$.

\emph{1) Bearmforming Design:} To address the first challenge, zero-forcing (ZF) beamforming $\mathbf{w}_i=\sqrt{p_i}\,\mathrm{e}^{\mathrm{j}\phi_i}\mathbf{w}_i^{\mathrm{ZF}}$ is adopted \cite{zf},
where $p_i$ and $\phi_i$ are the transmit power and the phase shift of the $i$-th beam, respectively, and $\mathbf{w}_{i}^{\mathrm{ZF}}$ is the steering vector for interference cancelation.
In particular, {to mitigate the phase differences among the signals $\{(1-x_{i})\mathbf{f}_i^{H}\mathbf{w}_i,\forall i\}$ in \eqref{drate}, we have $\phi_i=\angle(\mathbf{f}_i^{H}\mathbf{w}_i^{\mathrm{ZF}})$.}
On the other hand, as for $\mathbf{w}_{i}^{\mathrm{ZF}}$, define $\mathbf{F}_i=\left[\mathbf{g}_{ii}~\mathbf{h}_{i1}\dots\mathbf{f}_i\dots\mathbf{h}_{iK}\right]^H$ and
$\mathbf{H}_i=\left[\mathbf{h}_{i1}\dots\mathbf{f}_i\dots\mathbf{h}_{iK}\right]^H$ ($\mathbf{f}_i^H$ is at the
($i+1$)-th row of $\mathbf{F}_i$ and $i$-th row of $\mathbf{H}_i$).
According to \cite{zf}, the vector $\mathbf{w}_{i}^{\mathrm{ZF}}$ is the normalized $1$-st column of $\mathbf{F}_i^{H}\left(\mathbf{F}_i\mathbf{F}_i^{H}\right)^{-1}$ if $x_i = 1$ or the normalized $i$-st column of $\mathbf{H}_i^{H}\left(\mathbf{H}_i\mathbf{H}_i^{H}\right)^{-1}$ if $x_i = 0$.
Note that although the ZF beamforming is generally suboptimal, the gap between the proposed and optimal schemes is negligible for massive MIMO {implementations} due to the law of large numbers (the gap $\rightarrow 0$ when $M\rightarrow\infty$).
Putting $\mathbf{w}_i=\sqrt{p_i}\,\mathrm{e}^{\mathrm{j}\phi_i}\mathbf{w}_i^{\mathrm{ZF}}$ into problem $\mathcal{P}1$, $\mathcal{P}1$ is approximately converted to
\vspace{-0.2cm}
\begin{subequations}
	\begin{align}
	\mathcal{P}2:
	\mathop{\mathrm{max}}_{\substack{\mathbf{x},\{p_{i}\},\mathcal{S}}}~&
(1-\mu)\max_n\Theta(\mathcal{S}|n)+\mu R(\mathbf{x},\{p_i\})\nonumber\\
	\mathrm{s.t.}~~~ \label{P01}
	&x_i \in \left\{0,1\right\},~\forall i,\quad \mathcal{S}=\{i:x_i=1\}, \\
	&\sum_{i=1}^{K}p_i\leq P_{\rm{sum}},~0\leq p_i\leq P_T,~\forall i,\label{P02} \\
	&p_{i}|\mathbf{g}_{ii}^{H}\mathbf{w}_i^{\mathrm{ZF}}|^2
	\geq \sigma^2\gamma,~\forall i\in\mathcal{S}. \label{P03}
	\end{align}
\end{subequations}
where the simplified data-rate is
\begin{align}
R(\mathbf{x},\{p_i\})=\log_{2}\left[1+\frac{1}{\sigma^2}\left(\sum\limits_{i=1}^K|\mathbf{f}_i^{H}\mathbf{w}_i^{\mathrm{ZF}}|(1-x_{i})\sqrt{p_i}\right)^2
	\right]. \nonumber
\end{align}

\emph{2) Surrogate Fusion Function:} To address the second challenge, we propose to approximate the actual fusion accuracy through a binomial approximation \cite{Binomial}:
\begin{align}
	\hat{\Theta}(\mathcal{S}|n)&=\frac{1}{2}\sum\limits_{l=0}^{n-1}\tbinom{|\mathcal{S}|}{l}P^l(1-P)^{|\mathcal{S}|-l}
\nonumber\\
&\quad{}{}
+\frac{1}{2}\sum\limits_{l=0}^{|\mathcal{S}|-n}\tbinom{|\mathcal{S}|}{l}Q^{l}(1-Q)^{|\mathcal{S}|-l},
	\label{upper_bound}
\end{align}
where $P=\frac{1}{|\mathcal{S}|}\sum_{i=1}^{|\mathcal{S}|}P_i$ and $Q=\frac{1}{|\mathcal{S}|}\sum_{i=1}^{|\mathcal{S}|}Q_i$.
As such, the factorials are averaged out in \eqref{upper_bound}.
The following proposition is established to quantify the approximation gap and derive the associated voting threshold $n$.
\begin{proposition}
(i) The approximation gap is bounded as
\begin{align}
&\frac{1}{2}\sum_{n=1}^{|\mathcal{S}|}\left|\Theta(\mathcal{S}|n)-
\hat{\Theta}(\mathcal{S}|n)\right|
\nonumber\\
&\quad{}{}
\leq\mathcal{D}\left(\left\{P_i\right\}_{i=1}^{|\mathcal{S}|}\right)+
	\mathcal{D}\left(\left\{Q_i\right\}_{i=1}^{|\mathcal{S}|}\right),~\forall\mathcal{S},
\end{align}
where
$\mathcal{D}\left(\left\{P_i\right\}_{i=1}^{|\mathcal{S}|}\right) =
\frac{|\mathcal{S}|\left(1-P^{|\mathcal{S}|+1}-(1-P)^{|\mathcal{S}|+1}\right)}{(|\mathcal{S}|+1)P(1-P)}\sum_{i=1}^{|\mathcal{S}|}(P_i-P)^2$ and vice versa for $\mathcal{D}\left(\left\{Q_i\right\}_{i=1}^{|\mathcal{S}|}\right)$.

\noindent(ii) The optimal $n$ that maximizes $\hat{\Theta}(\mathcal{S}|n)$ is
\vspace{-0.1cm}
\begin{equation}
	n^\diamond = \min\left(|\mathcal{S}|,\left\lceil \frac{|\mathcal{S}|}{1+\alpha} \right\rceil\right), \label{n_opt}
\end{equation}
\vspace{-0.1cm}
where $\alpha =\frac{\ln\frac{P}{1-Q}}{\ln\frac{Q}{1-P}}$ and $\lceil \cdot \rceil$ denotes the ceiling function.
\end{proposition}
\begin{proof}
Part (i) is proved based on the Poisson approximation theorem \cite{poisson} and Chebyshev's inequality.
Part (ii) is proved by setting the derivative $\partial\hat{\Theta}/\partial n$ to zero.
For more details please refer to Appendix A.
\end{proof}

\begin{figure*}[!t]
	\centering
		\includegraphics[width=0.98\textwidth]{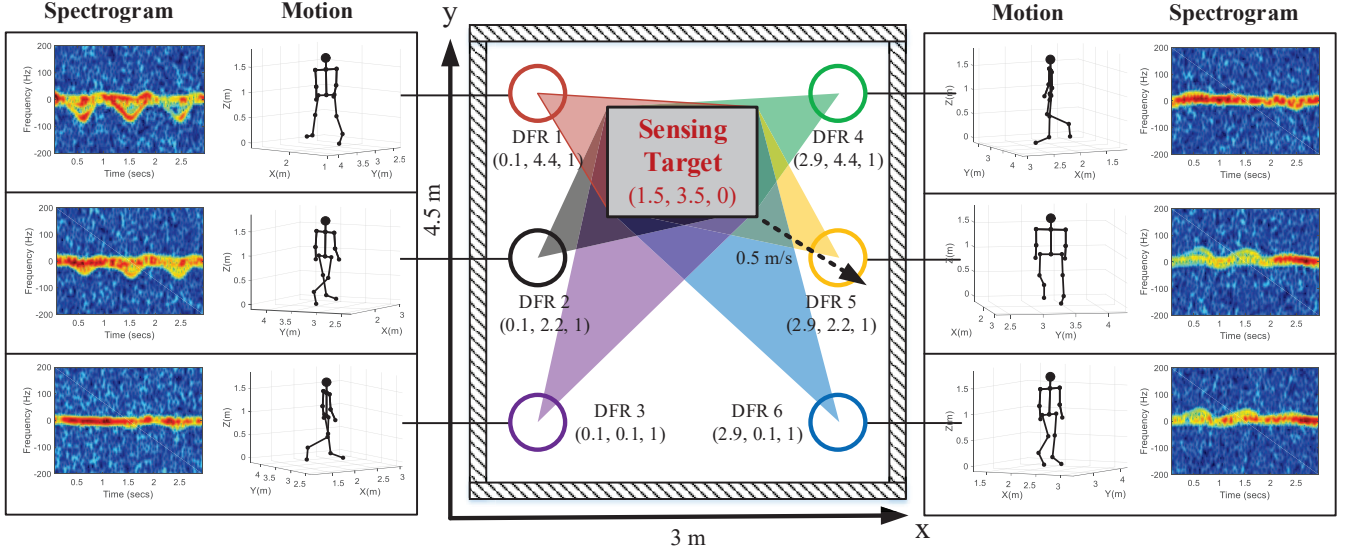}
		\vspace{-0.48cm}
	\caption{Illustration of the simulated scenario and the multi-view sensing data at DFRs. The associated false negative and false positive rates are given by $\left\{P_i\right\}_{i=1}^{6} = \{0.05,0.09,0.12,0.14,0.05,0.23\}$ and $\left\{Q_i\right\}_{i=1}^{6} = \{0.09,0.14,0.07,0.16,0.05,0.03\}$, respectively.
} \label{Simulation}
\vspace{-6pt}
\end{figure*}

\textbf{Proposition 1} states that the approximation model is close to the actual fusion accuracy.
To illustrate this, consider the case of $|\mathcal{S}|\,=\,7$ with $\{P_i,Q_i\}$ specified in Fig.~\ref{Optimla_n}.
It can be seen from Fig.~\ref{Optimla_n} that the gap between the approximate and actual fusion accuracy is nearly negligible.
Furthermore, \textbf{Proposition 1} states that the naive majority voting is generally not optimal.
This can be seen from Fig.~\ref{Optimla_n} that the optimal $n^\diamond$ is $3$ out of $7$ rather than $4$ out of $7$
and the theoretical voting threshold in \textbf{Proposition 1} matches the experimental result in Fig.~\ref{Optimla_n}.
Finally, \textbf{Proposition 1} states that $\hat{\Theta}$ is a monotonically increasing function of $|\mathcal{S}|$.
Therefore, the fusion gain $\hat{\Theta}(\mathcal{S}|n^\diamond)-\frac{1}{2}(P+Q)$ is always positive.

Based on \textbf{Proposition 1}, problem $\mathcal{P}2$ is converted into
\begin{align}
	\mathcal{P}3:
	\mathop{\mathrm{max}}_{\substack{\mathbf{x},\{p_{i}\},\mathcal{S}}}~&
\Xi\left(\mathbf{x},\{p_{i}\},\mathcal{S}\right), \quad \mathrm{s.t.}~\eqref{P01}-\eqref{P03},
\end{align}
where the new surrogate objective function is
\begin{align}
\Xi\left(\mathbf{x},\{p_{i}\},\mathcal{S}\right)&=(1-\mu)\hat{\Theta}\left[\mathcal{S}|\min\left(|\mathcal{S}|,\left\lceil \frac{|\mathcal{S}|}{1+\alpha} \right\rceil\right)\right]
\nonumber\\
&\quad{}{}+\mu R(\mathbf{x},\{p_i\}).
\end{align}

\emph{3) HMO Algorithm:}
To address the third challenge,
we present a computationally efficient algorithm to solve problem $\mathcal{P}3$.
In particular, we leverage the HMO framework in \cite{hybrid}, which starts from a feasible solution of $\mathbf{x}$ (e.g., $\mathbf{x}^{[0]}=[1,0,\cdots,0]^{T}$),
and randomly selects a candidate solution $\mathbf{x}'$ from the neighborhood
\begin{align}
\mathcal{N}(\mathbf{x}^{[0]})=\{\mathbf{x}:||\mathbf{x}-\mathbf{x}^{[0]}||_0\leq L,~x_i\in\{0,1\}, \forall i\},
\end{align}
where $L\geq 1$ is the variable size of neighborhood.
It can be seen that $\mathcal{N}(\mathbf{x}^{[0]})$ is a subset of the entire feasible space containing solutions ``close'' to $\mathbf{x}^{[0]}$.
$\mathcal{N}(\mathbf{x}^{[0]})$ is generated by randomly flipping $L$ elements inside $\{x_i\}$ \cite{hybrid}.
With the neighborhood $\mathcal{N}(\mathbf{x}^{[0]})$ defined above and the choice of $\mathbf{x}$ fixed to $\mathbf{x}=\mathbf{x}'\in\mathcal{N}(\mathbf{x}^{[0]})$,
the set of sensing DFRs is $\mathcal{S}'=\{i:x_i'=1\}$.
Then the problem $\mathcal{P}3$ w.r.t. $\{p_i\}$ keeping $\{\mathbf{x}=\mathbf{x}',\mathcal{S}=\mathcal{S}'\}$ fixed is
\begin{subequations}
	\begin{align}
	\mathcal{P}4: \mathop{\mathrm{max}}_{\substack{\{p_i\}}}~&\sum\limits_{i=1}^K(1-x_{i}')|\mathbf{f}_i^{H}\mathbf{w}_i^{\mathrm{ZF}}|\sqrt{p_i} \nonumber\\
	\mathrm{s.t.}~~~ &\sum_{i=1}^{K}p_i\leq P_{\rm{sum}},\quad0\leq p_i\leq P_T,~\forall i,\label{P201} \\
	&p_{i}|\mathbf{g}_{ii}^{H}\mathbf{w}_i^{\mathrm{ZF}}|^2\geq \sigma^2\gamma,~\forall i\in\mathcal{S}', \label{P202}
	\end{align}
\end{subequations}
where we have removed the terms not related to $\{p_i\}$ and the logarithm and quadratic functions due to their monotonicity.
Since the objective function is concave in $\{p_i\}$ and the constraints are linear, the problem $\mathcal{P}4$ is a convex optimization problem w.r.t. $\{p_{i}\}$, which can be readily solved via the open-source software CVXPY (https://www.cvxpy.org/index.html) with a computational complexity of $\mathcal{O}(K^{3.5})$.
Let $\{p_i'\}$ denote the optimal solution of $\{p_i\}$ to $\mathcal{P}4$.
We consider two cases:
(i) If $\Xi(\mathbf{x}',\{p_{i}'\},\mathcal{S}')\geq\Xi(\mathbf{x}^{[0]},\{p_{i}^{[0]}\},\mathcal{S}^{[0]})$, we update $\mathbf{x}^{[1]}\leftarrow\mathbf{x}'$. By treating $\mathbf{x}^{[1]}$ as a new feasible solution, we can construct the next neighborhood $\mathcal{N}(\mathbf{x}^{[1]})$;
(ii) If $\Xi(\mathbf{x}',\{p_{i}'\},\mathcal{S}')<\Xi(\mathbf{x}^{[0]},\{p_{i}^{[0]}\},\mathcal{S}^{[0]})$, we re-generate another point within the neighborhood $\mathcal{N}(\mathbf{x}^{[0]})$ until $\Xi(\mathbf{x}',\{p_{i}'\},\mathcal{S}')\geq\Xi(\mathbf{x}^{[0]},\{p_{i}^{[0]}\},\mathcal{S}^{[0]})$.

The above procedure is repeated to generate a sequence of $\{\mathbf{x}^{[1]},\mathbf{x}^{[2]},\cdots\}$ and the converged point is guaranteed to be a local optimal solution to $\mathcal{P}3$ \cite{dimitri}.
Moreover, as shown in Fig.~\ref{converge}, the proposed HMO method achieves performance close to that of the optimal solution obtained by exhaustive search.
In practice, we can terminate the iterative procedure when the number of iterations is larger than $\overline{\mathrm{Iter}}$, e.g., we can set $\overline{\mathrm{Iter}}=10$ for Fig.~\ref{converge}.
The complexity of the HMO algorithm is thus $\mathcal{O}(\overline{\mathrm{Iter}}\,K^{3.5})$.

\section{Simulation Results}

\begin{figure*}[!t]
	\centering
	\subfigure[]{
		\label{Comp} %% label for second subfigure 63 43
		\includegraphics[height=2.9in]{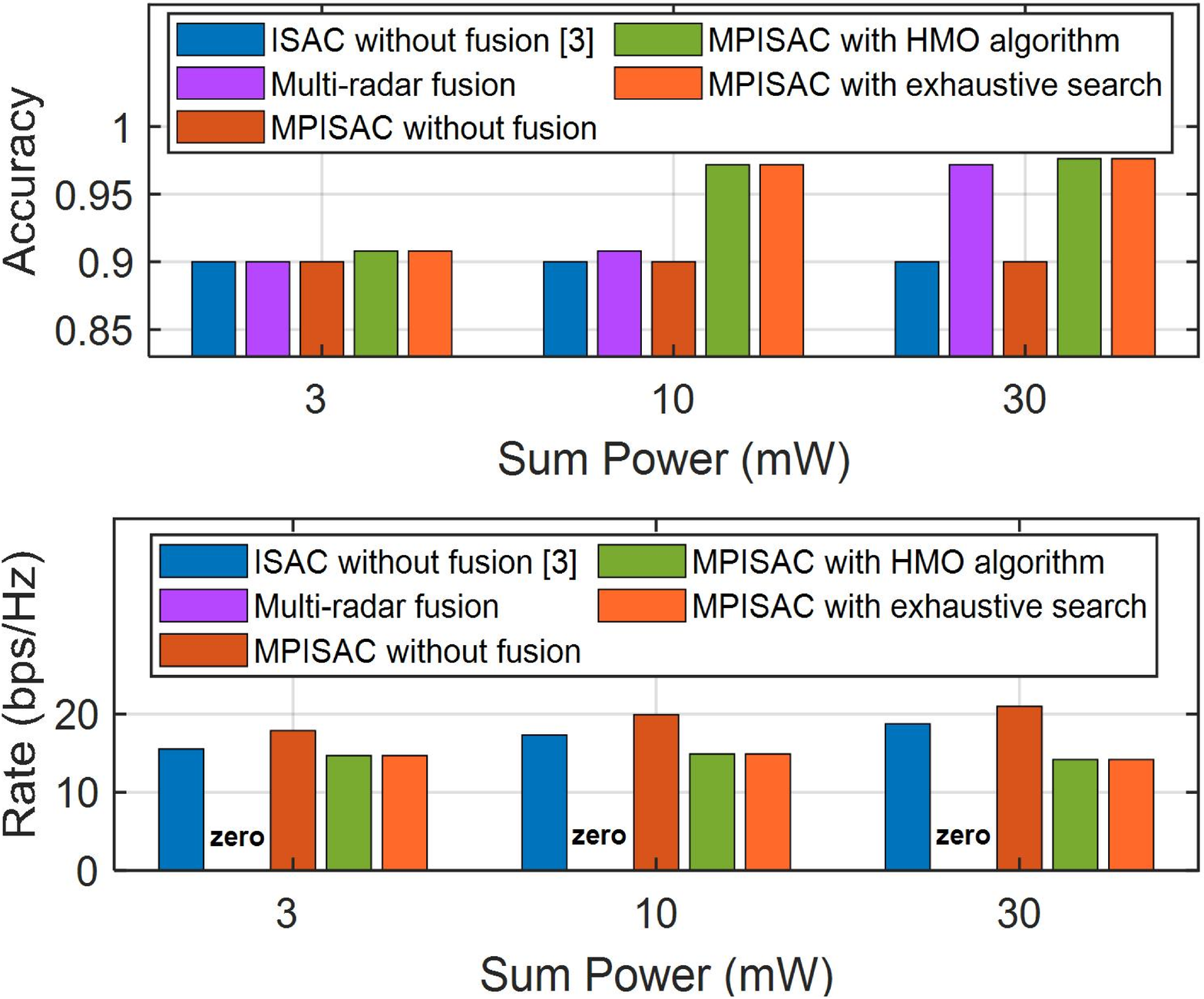}}
	\subfigure[]{
		\label{Region} %% label for second subfigure 63 43
		\includegraphics[height=2.9in]{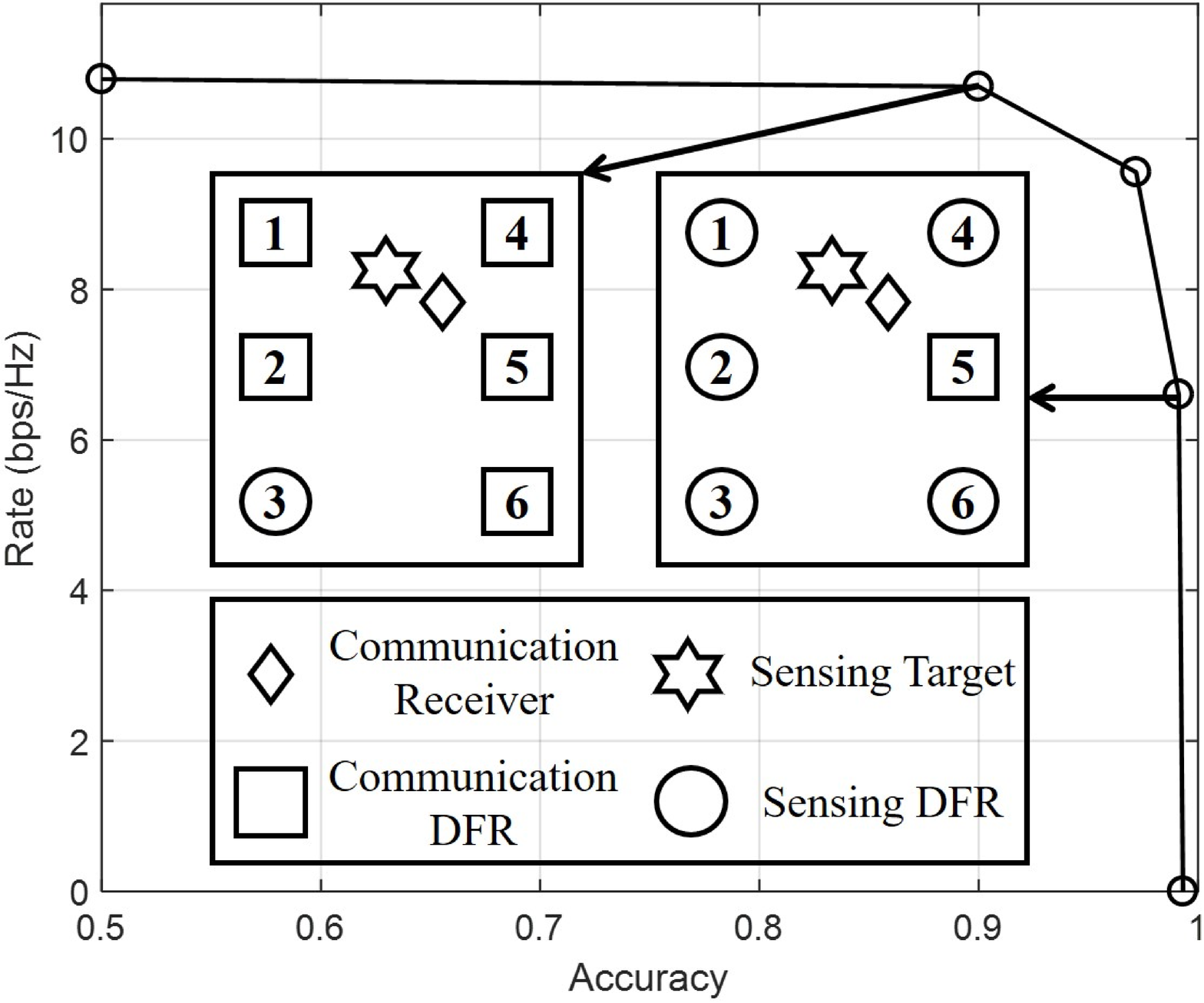}}
		\vspace{-0.3cm}
	\caption{(a) Comparison between the MPISAC, ISAC without fusion and multi-radar fusion;
	(b) Accuracy-rate region of the MPISAC system, where
the functionality results for the case of $|\mathcal{S}|=1$ and $|\mathcal{S}|=5$ are illustrated.}
	\label{Simulation_2} %% label for entire figure
	\vspace{-6pt}
\end{figure*}

This section provides simulation results to illustrate the performance of the MPISAC scheme.
We simulate the case of $K=6$ in a conference room with size $3\times4.5\times3\,\rm{m}^3$.
A wireless radar simulator \cite{SPAWC} based on ray-tracing is leveraged to generate the spectrogram datasets at multi-view DFRs.
The target is either an adult (i.e., normal) or a child (i.e., abnormal).
The maximum transmit power is $P_T=10\,$mW, the sum transmit power is $P_{\rm{sum}}=50\,$mW, and the noise power is $\sigma^2=-50\,$dBm.
The channels are generated by using a distance-dependent path-loss model \cite{channel}.

First, Fig.~\ref{Simulation} demonstrates the sensing data when an adult is walking toward the DFR 5 (i.e., yellow circle) with a speed of $0.5\,$m/s.
It can be seen that the spectrograms at DFRs 3 and 4 (i.e., purple and green circles) are flat, as the human motion direction is orthogonal to their observation angles.
On the other hand, the spectrograms at DFRs 1 and 5 (i.e., red and yellow circles) are the most fluctuated, as the human motion is parallel to their observation angles.
Note that DFRs 1 and 5 have opposite patterns as they face different sides of the human body.
This visualization indicates that fusing the detection results from DFRs 1 and 5 may help improve the performance at DFRs 3 and 4.
However, the exact fusion gain should be further quantified.
To this end, we compare the performance of the proposed MPISAC with that of ISAC without fusion \cite{Fan_2} (i.e., $|\mathcal{S}|=1$) in Fig.~\ref{Comp}.
It can be seen that, by introducing the fusion mechanism, MPISAC significantly improves the detection accuracy and the gap quantifies the \textbf{fusion gain} brought by the diversity from effectively exploiting the multi-view observations across different DFRs as shown in Fig.~\ref{Simulation}.
The fusion gain vanishes as the sum power decreases with fewer effective DFRs involved.
Note that the proposed MPISAC without fusion achieves higher data-rate while guaranteeing the same sensing accuracy compared with ISAC without fusion.
This performance gain is brought by the proposed functionality selection.

Next, we compare the performance of MPISAC with that of multi-radar fusion which follows the principle in \cite{Fusion_1} (i.e., $\mathbf{x}=[1,1,1,\cdots]^T$) in Fig.~\ref{Comp}.
The multi-radar fusion scheme leads to a zero data-rate as all the devices work in the sensing state.
Furthermore, its sensing accuracy is worse than that of MPISAC with partial DFRs for sensing.
This is because the DFR far from the target would consume excessive power resources.
Adding a remote DFR would reduce the SINRs of other DFRs, resulting in fewer effective DFRs.
Consequently, switching remote DFRs to the communication mode would provide communication and sensing gains simultaneously.
We define this gain as \textbf{ISAC gain} which is significant under stringent power budgets and negligible with sufficient resources.

Finally, the accuracy-rate trade-off region is given in Fig.~\ref{Region} by varying the value of $\mu$ in $[0,1]$.
First, increasing $|\mathcal{S}|$ results in a larger accuracy but a lower data-rate.
Second, at $|\mathcal{S}|=1$, DFR 3 (i.e., farthest from the receiver) instead of DFR 1 (i.e., closest to the target) is selected for sensing.
This is because communication is more important than sensing at this boundary point.
In order to achieve a high data-rate, it is necessary to select a DFR with the worst communication channel for sensing, which is DFR 3 in the considered scenario as shown in Fig.~\ref{Region}.
On the other hand, at $|\mathcal{S}|=5$, we select the DFR closest to the receiver (i.e., DFR 5) for communication and the remaining DFRs for sensing.
This is because the sensing accuracy is close to $1$ by fusing $5$ arbitrary DFR outputs.
Finally, by varying the value of $\mu$, it is possible to achieve a higher data-rate (the far left boundary point),
a higher sensing accuracy (the far right boundary point) and a more balanced rate-accuracy pair (the middle boundary point).

\vspace{-0.2cm}
\section{Conclusion}

This letter derived a set of models for MPISAC systems to quantify the accuracy of cooperative sensing and the rate of cooperative communication.
The optimized functionality selection was proposed and the fusion/ISAC gain was illustrated, which shows that adaptive
selection between sensing and communication is important for the effective exploration of the ISAC trade-off region.

\appendices
\section{Proof of Proposition 1} \label{appx: Proposition1}

To prove part (i), we define $N=|\mathcal{S}|$ and
\begin{align}
	\mathcal{L}(n,\left\{P_i\right\}_{i=1}^N)&=\frac{1}{2}\sum\limits_{l=0}^{n-1}\sum\limits_{F\in\mathcal{F}_l}\prod
	_{i\in F}P_i\prod_{i\in F^c}(1-P_i),\nonumber\\
	\mathcal{B}(n,P)&=\frac{1}{2}\sum\limits_{l=0}^{n-1}\tbinom{N}{l}P^l(1-P)^{N-l}. \nonumber
\end{align}
As such, the primal and approximate fusion accuracy functions are written as
$\Theta = \mathcal{L}(n,\left\{P_i\right\}_{i=1}^N) + \mathcal{L}(N-n+1,\left\{Q_i\right\}_{i=1}^N)$ and
$\hat{\Theta} = \mathcal{B}(n,P) + \mathcal{B}(N-n+1,Q)$, respectively.
Based on the Poisson approximation theorem \cite{poisson} and
Chebyshev's inequality, we have \cite{Binomial}
\begin{align}
	\frac{1}{2}\sum_{n=1}^{N}\left|\mathcal{L}(n,\left\{P_i\right\}_{i=1}^N)-
	\mathcal{B}(n,P)\right|
    \leq\mathcal{D}\left(\left\{P_i\right\}_{i=1}^N\right).
\end{align}
Furthermore, according to the Triangle inequality, we have
\begin{align}
	\frac{1}{2}\sum_{n=1}^{N}\left|\Theta-
\hat{\Theta}\right|
&\leq\frac{1}{2}\sum_{n=1}^{N}\left|\mathcal{L}(n,\left\{P_i\right\}_{i=1}^N)-
\mathcal{B}(n,P)\right|
\nonumber\\
&\quad{}{}
+\frac{1}{2}\sum_{n=1}^{N}\left|\mathcal{L}(n,\left\{Q_i\right\}_{i=1}^N)-
\mathcal{B}(n,Q)\right|\nonumber
\\&\leq\mathcal{D}\left(\left\{P_i\right\}_{i=1}^N\right)+
	\mathcal{D}\left(\left\{Q_i\right\}_{i=1}^N\right).
\end{align}
This completes the proof for part (i).

To prove part (ii), we need to solve the optimal voting problem
	\begin{align}
	\mathop{\mathrm{max}}_{n}\quad&	\hat{\Theta}(\mathcal{S}|n), \quad
	\mathrm{s.t.}~~n\in[0,\cdots,N].	\label{voting}
	\end{align}
First, the formula of $\hat{\Theta}$ is transformed as
\begin{equation}
\begin{aligned}	
	\hat{\Theta} &= \frac{1}{2}\sum\limits_{l=0}^{n-1}\tbinom{N}{l}P^l(1-P)^{N-l} +
	\frac{1}{2}\sum\limits_{l=0}^{N-n}\tbinom{N}{l}Q^l(1-Q)^{N-l}\\
	% &=\frac{1}{2}\sum\limits_{l=0}^{n-1}\tbinom{N}{l}P^l(1-P)^{N-l}
	% + \frac{1}{2}\left[1-\sum\limits_{l=0}^{n-1}\tbinom{N}{l}Q^{N-l}(1-Q)^l\right]\\
	&=\frac{1}{2}+\frac{1}{2}\sum\limits_{l=0}^{n-1}\tbinom{N}{l}[P^l(1-P)^{N-l}-Q^{N-l}(1-Q)^l]. \label{transformed}
\end{aligned}
\end{equation}	
Next, the derivative of $\hat{\Theta}(\mathcal{S}|n)$ w.r.t. $n$ is
\begin{equation}
\begin{aligned}	
	\frac{\partial \hat{\Theta}}{\partial n} &\approx \hat{\Theta}(\mathcal{S}|n+1)-\hat{\Theta}(\mathcal{S}|n)\\
	&= \frac{1}{2}\tbinom{N}{n}[P^n(1-P)^{N-n}-Q^{N-n}(1-Q)^n]. \label{process_1}
\end{aligned}
\end{equation}	
Third, the optimal value of $n$ is obtained when $\frac{\partial \hat{\Theta}}{\partial n}=0$:
\begin{equation}
	\left[P^n(1-P)^{N-n}-Q^{N-n}(1-Q)^n\right]=0. \label{process_2}
\end{equation}
Taking the logarithm on the both side yields
\begin{equation}
	n^\star = \frac{N\left[\ln Q-\ln(1-P)\right]}{\ln P-\ln(1-P)+\ln Q-\ln(1-Q)}. \label{n}
\end{equation}
Finally, denote $\alpha = \frac{\ln\frac{P}{1-Q}}{\ln\frac{Q}{1-P}}$ and round up $n^\star$.
We obtain $n^\diamond\approx \left \lceil \frac{N}{1+\alpha} \right \rceil$.
The proof is thus completed.


\vspace{-0.3cm}
\begin{thebibliography}{50}

\bibitem{ISAC} J.~A.~Zhang \emph{et al.}, ``An overview of signal processing techniques for joint communication and radar sensing,'' \emph{IEEE J. Sel. Topics Signal Process.}, vol. 69, no. 2, pp. 1295--1315, Nov. 2021.

\bibitem{Fan_1} Q.~Zhang \emph{et al.}, ``Time-division ISAC enabled connected automated vehicles: Cooperation algorithm design and performance evaluation,'' \emph{IEEE J. Sel. Areas Commun.}, vol. 40, no. 7, pp. 2206--2218, Jul. 2022.

\bibitem{Fan_2} F.~Liu \emph{et al.}, ``Cram\'er-Rao bound optimization for joint radar-communication beamforming,'' \emph{IEEE Trans. Signal Process.}, vol. 70, pp. 240--253, Dec. 2021.

\bibitem{Noise} S.~Z.~Gurbuz and M.~G.~Amin, ``Radar-based human-motion recognition with deep learning: Promising applications for indoor monitoring,'' \emph{IEEE Signal Process. Mag.}, vol. 36, no. 4, pp. 16--28, Jul. 2019.

\bibitem{Fusion_1} W.~Yi and L.~Chai, ``Heterogeneous multi-sensor fusion with random finite set multi-object densities,'' \emph{IEEE Trans. Signal Process.}, vol. 69, pp. 3399-3414, Jun. 2021.

\bibitem{Fusion_2} A.~Aguileta \emph{et al.}, ``Multi-sensor fusion for activity recognition: A survey,'' \emph{Sensors}, vol. 19, no. 17, Sep. 2019.

\bibitem{SPAWC} G.~Li \emph{et al.}, ``Wireless sensing with deep spectrogram network and primitivebased autoregressive hybrid channel model,'' in \emph{Proc. IEEE SPAWC}, Lucca, Italy, Sep. 2021, pp. 481--485.

\bibitem{voting} P.~K.~Varshney, \emph{Distributed Detection and Data Fusion}. New York: Springer-Verlag, 1997.

\bibitem{zf} F.~Liu \emph{et al.}, ``Integrated sensing and communications: Toward dual-functional wireless networks for 6G and beyond,'' \emph{IEEE J. Sel. Areas Commun.}, vol. 40, no. 6, pp. 1728--1767, Jun. 2022.

\bibitem{Weight} R.~T.~Marler \emph{et al.}, ``Survey of multi-objective optimization methods for engineering,'' \emph{Struct. Multidiscipl. Optim.}, vol. 26, pp. 369--395, Apr. 2004.

\bibitem{Binomial} W.~Ehm, ``Binomial approximation to the Poisson binomial distribution,'' \emph{Stat. Probab. Lett.}, vol. 11, no. 1, Jan. 1991.

\bibitem{hybrid} E. G. Talbi, ``Combining metaheuristics with mathematical programming, constraint programming and machine learning,'' \emph{Ann. Oper. Res.}, vol. 240, no. 1, pp. 171-215, May 2016.

\bibitem{dimitri} D. P. Bertsekas, \emph{Network Optimization: Continuous and Discrete Models}. Athena Scientific, 1998.

\bibitem{channel} A.~Goldsmith, \emph{Wireless Communications}. New York: Cambridge University Press, 2005.

\bibitem{poisson} A.~D.~Barbour and G.~K.~Eagleson ``Poisson approximation for some statistics based on exchangeable trials,'' \emph{Adv. in Appl. Probab.}, vol. 15, no. 3, pp. 585--600, Sep. 1983.

\end{thebibliography}
\end{document}